\documentclass[letterpaper, 10 pt, conference]{ieeeconf} 

\IEEEoverridecommandlockouts 

\usepackage{amssymb, amsthm, amsmath, latexsym, amsfonts, mathrsfs, amsbsy,mathtools,relsize}
\usepackage[T1]{fontenc}
\usepackage[dvips]{graphicx}
\usepackage[x11names]{xcolor}
\usepackage{listings}
\usepackage{enumerate}
\usepackage{flushend}
\usepackage{soul}
\usepackage{pgfplots}
\usepackage{tikz}
\usepackage{float}
\usepackage{dsfont}
\usepackage{verbatim}
\usepackage{algorithm}
\usepackage[algo2e]{algorithm2e} 
\usepackage[noend]{algpseudocode}
\usepackage[colorlinks]{hyperref}
\hypersetup{
    citecolor=Firebrick4,
    linkcolor=Firebrick4,   
    urlcolor=Red3}
\usepackage{cite}

\usetikzlibrary{shapes,shapes.geometric,arrows,arrows.meta,fit,calc,positioning,automata,through,intersections}
\tikzset{diamond state/.style={draw,diamond}}

\newtheoremstyle{theoremdd}
  {\topsep}
  {\topsep}
  {\itshape}
  {0pt}
  {\bfseries}
  {.}
  { }
  {\thmname{#1}\thmnumber{ #2}\textnormal{\thmnote{ (#3)}}}

\theoremstyle{theoremdd}
\newtheorem{theorem}{Theorem}
\newtheorem{lemma}{Lemma}
\newtheorem{assumption}{Assumption}

\newtheorem{remark}{Remark}

\newcommand{\bm}[1]{\boldsymbol{#1}} 
\newcommand{\set}[1]{\mathcal{#1}} 
\newcommand{\ie}{\textit{i.e.,~}} 
\newcommand{\eg}{\textit{e.g.,~}} 
\newcommand{\inneighbor}[1]{\set{N}_{#1}^{\texttt{in}}}
\newcommand{\outneighbor}[1]{\set{N}_{#1}^{\texttt{out}}}
\newcommand{\indegree}[1]{d_{#1}^{\texttt{in}}}
\newcommand{\outdegree}[1]{d_{#1}^{\texttt{out}}}

\newcommand{\ouralgorithm}{{\fontsize{10.5pt}{10.5pt}\selectfont \textsc{pp-acdc}}}

\title{\LARGE \bf
{Average Consensus with Dynamic Compression\\in Bandwidth-Limited Directed Networks}}

\author{
Evagoras Makridis$^{1,*}$, Gabriele Oliva$^{2}$, Apostolos I. Rikos$^{3}$, and Themistoklis Charalambous$^{1,4}$
\thanks{$^1$Department of Electrical and Computer Engineering, School of Engineering, University of Cyprus, Nicosia, Cyprus.}
\thanks{$^2$Department of Engineering, University Campus Bio-Medico of Rome, Via Alvaro del Portillo, 21 - 00128 Roma, Italy.} 
\thanks{$^3$Artificial Intelligence Thrust of the Information Hub, The Hong Kong University of Science and Technology (Guangzhou), Guangzhou, China, and Department of Computer Science and Engineering, The Hong Kong University of Science and Technology, Clear Water Bay, Hong Kong. E-mail: {\tt~apostolosr@hkust-gz.edu.cn}.}
\thanks{$^4$Department of Electrical Engineering and Automation, School of Electrical Engineering, Aalto University, Espoo, Finland.}
\thanks{$^*$Corresponding author. Email: {\tt makridis.evagoras@ucy.ac.cy}.}
\thanks{This work has been partly funded by MINERVA, a European Research Council (ERC) project funded under the European Union's Horizon 2022 research and innovation programme (Grant agreement No. 101044629).}
}

\begin{document}

\maketitle
\thispagestyle{empty}
\pagestyle{empty}

\begin{abstract}
In this paper, the average consensus problem has been
considered for directed unbalanced networks under finite bit-rate communication. We propose the Push-Pull Average Consensus algorithm with Dynamic Compression (\ouralgorithm) algorithm, a distributed consensus algorithm that deploys an adaptive quantization scheme and achieves convergence to the exact average without the need of global information. 
A preliminary numerical convergence analysis and simulation results corroborate the performance of \ouralgorithm.
\end{abstract}


\section{Introduction}\label{sec:introduction}

Distributed consensus algorithms are essential components of modern applications that involve the cooperation of a group of networked agents (often referred to as \emph{multi-agent systems}) such as robotic networks, distributed power networks, and wireless sensor networks. Some of these applications include distributed optimization \cite{ carnevale2023distributed,maritan2024fully}, distributed estimation and control \cite{ makridis2024fully, fioravanti2024distributed}, and distributed learning \cite{bastianello2024robust}. In such problems, the agents in a network aim at reaching agreement on a common decision by iteratively updating their values based on information received by their immediate neighbors via information exchange. When the agreed common value is the average of the initial values of all agents in the network, we say that the agents reach \emph{average consensus}. Although there exist several average consensus schemes considering network abnormalities (such as packet delays or drops)~\cite{SYS-016} and malfunctioning agents (e.g., curious, malicious, or faulty agents)~\cite{CHARALAMBOUS2024}, they require agents to be able to send and receive real values with infinite precision.


In multi-agent systems where the communication is often established over wireless channels, several constraints, such as limited bandwidth, power, and memory, arise. Such limitations often require that the information that is exchanged between agents be quantized. Quantization enables information compression such that a number in a continuous set (\ie real-valued number $x\in\mathbb{R}$), can be mapped to a value in a finite set. This mapping is done with a quantizer, denoted by $Q(\cdot)$. In this work, we define the quantizer given a positive integer number of bits $b<\infty$ available for quantization, the step-size $\Delta$ which defines the spacing between the quantization intervals, and the quantizer's symmetry point (or midpoint) $\sigma$ (which is not necessarily at $0$) as follows:
\begin{align}\label{eq:quantizer}
    Q\big(x,b,\Delta,\sigma\big) &= \begin{cases}
        \hfil \sigma + \left(2^{b-1} - 1\right)\Delta, & \text{if}~x\geq\sigma + \bar{x},\\
        \hfil \sigma - \left(2^{b-1} - 1\right)\Delta, & \text{if}~x\leq\sigma - \bar{x},\\
        \hfil \sigma + \Delta \bigl \lfloor \frac{x-\sigma}{\Delta} \bigr \rceil, & \text{otherwise},
    \end{cases}
\end{align} 
where the quantization range limit is given by
\begin{align}
\bar{x}={\left(2^{b-1} - \dfrac{1}{2}\right)\Delta},
\end{align}
and $\lfloor x \rceil$ represents the round function for a value $x$. 
An example of the input/output relation of a quantizer with $b=3$ and $\Delta=1$ is shown in Fig.~\!\ref{fig:quantizer}.
\begin{figure}[t]
    \centering
    \includegraphics[width=0.75\linewidth]{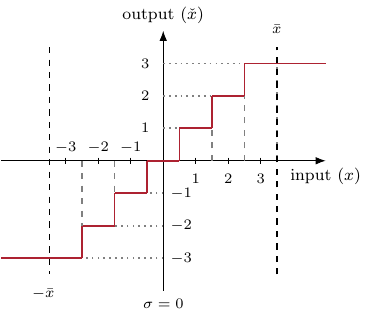}
    \caption{A uniform quantizer of $b=3$ bits with $\sigma=0$ and $\Delta=1$.}
    \label{fig:quantizer}
\end{figure}
Although quantization reduces the precision of the original information, it enables mapping to a finite set of discrete values representable using bits. This mapping is necessary for transmissions over channels with limited bandwidth and transmitters with power limitations. Bandwidth limits the number of symbols sent per second, while power constraints limit the number of bits per symbol that the modulation scheme can reliably encode. Increasing the number of bits per symbol requires higher transmission power to maintain reliable decoding at the receiver in the presence of noise. Consequently, quantization encodes real-valued information into sequence of bits aligned with the modulation scheme and constrained by the channel capacity.

Communicating quantized values in multi-agent systems has attracted significant interest due its essential role in overcoming practical constraints like bandwidth and power limitations \cite{liu2024distributed,lin2024event,yan2025layered,doostmohammadian2025log,feng2025quantized}. However, in the context of distributed average consensus, quantized communication can negatively impact both the accuracy and overall performance of the system. Specifically, it affects the system behavior as follows. First, saturation occurs when the signal exceeds the quantization range, leading to large errors that can destabilize the system, if the consensus algorithm was originally designed for a non-quantized system. Second, as the system state gets closer to its target value, performance degradation arises because higher precision is required. Finally, each time an agent quantizes a value before sending it to its neighbors, some information is lost due to rounding and discretization. Over time, these small errors can accumulate, leading to biased consensus, \ie agents converge to a value that is far from the true average. These errors are even amplified when agents interact over a directed network, where the information flow is not balanced, since the errors may propagate asymmetrically within the network. Since this work focuses on average consensus in directed graphs, the discussion of related work will be limited to literature specifically addressing possibly unbalanced directed graphs.

To minimize quantization errors, agents can use a dynamic quantization strategy, such as zooming in / zooming out (ZIZO) \cite{Brockett2000TAC}, where the quantization step size changes based on the range of values being exchanged. However, since the links are not bidirectional, it is difficult to synchronize the nodes to simultaneously update the size of the quantization step. 
The majority of algorithms designed for quantized average consensus for unbalanced directed graphs achieve convergence in a probabilistic manner (see, \eg~\cite{Rikos2021TAC} and references therein), meaning that nodes reach consensus on the quantized average with probability one or through other probabilistic guarantees. Nevertheless, either residual quantization errors persist indefinitely or the average consensus is approximate, the convergence is slow and depends on the network size. With only a few exceptions~\cite{Zambieri2009ECC,song2022compressed,rikos2024finite}, the development of deterministic distributed strategies for achieving quantized average consensus in unbalanced directed graphs remains largely unexplored. Baldan and Zambieri~\cite{Zambieri2009ECC} proposed distributed algorithms that achieve average consensus with quantized transmissions and the algorithm's parameters can be chosen in a distributed fashion without knowing the number of agents composing the network. Since they only adopt the ZIZO strategy without any adaptation on the symmetry point $\sigma$ of the quantizer, there is a minimum number of bits required for the algorithm to converge which depends on the parameter choice. Rikos \emph{et al.}~\cite{rikos2024finite} proposed a distributed algorithm that is able to calculate the exact average of the initially quantized values (in the form of a quantized fraction) in a deterministic fashion after a finite number of iterations. However, the average consensus is approximate (due to the initial quantization), the number of bits for each message is not necessarily bounded by a certain value, and the convergence is slow and depends on the network size. 

In this paper, our aim is to address the aforementioned challenges and thus we propose a distributed average consensus algorithm for possibly unbalanced directed graphs under limited communication bandwidth. The algorithm integrates an adaptive quantization mechanism, enabling agents to converge to the \emph{exact} average without requiring global knowledge of the network (see \S\ref{sec:algo}). Our algorithm adopts a surplus consensus structure suited for unbalanced directed graphs and avoids the division operations used in ratio consensus methods (see \cite{SYS-016}), making it robust to quantization. We introduce a design that integrates dynamic quantization (zooming and midpoint shifting) directly into the consensus updates, ensuring robustness to finite-bit communication constraints and enabling the consensus error to be driven to zero. We analyze the convergence properties of the proposed method (see \S\ref{sec:algo}) and support our claims with simulations (see \S\ref{sec:simulations}) that demonstrate exact average consensus under tight communication constraints.

\section{Preliminaries}\label{sec:background}

\subsection{Mathematical Notation}
The sets of real, integer, and natural numbers are denoted as $\mathbb{R}$, $\mathbb{Z}$, and $\mathbb{N}$, respectively. The set of nonnegative integer numbers and the set of nonnegative real numbers are denoted as $\mathbb{Z}_{\geq0}, \mathbb{R}_{\geq 0}$. The nonnegative orthant of the $n$-dimensional real space $\mathbb{R}^n$ is denoted as $\mathbb{R}_{\geq 0}^{n}$. 
Matrices are denoted by capital letters, and vectors by small letters. The transpose of a matrix $A$ and a vector $x$ are denoted as $A^\top$, $x^\top$, respectively. The all-ones and all-zeros vectors are denoted by $\mathbf{1}$ and $\mathbf{0}$, respectively, with their dimensions being inferred from the context. The floor function of a real number $x$ is denoted as $\lfloor x \rfloor = \max \{ b \in \mathbb{Z} \mid b \le x \}$, and is defined as the greatest integer less than or equal to $x$. The ceiling function of a real number $x$ is denoted as $\lceil x \rceil = \min \{ b \in \mathbb{Z} \mid b \ge x \}$, and is defined as the least integer greater than or equal to $x$. 

\subsection{Network Model}
Consider a group of $n>1$ agents communicating over a time-invariant directed network. The interconnection topology of the communication network is modeled by a digraph $\set{G}=(\set{V}, \set{E})$. Each agent $v_j$ is included in the set of digraph nodes $\set{V}=\{v_1, \cdots, v_n\}$, whose cardinality is $n=\left|\set{V}\right|$. The interactions between agents are included in the set of digraph edges $\set{E} \subseteq \set{V} \times \set{V}$. The total number of edges in the network is denoted by $m=\left|\set{E}\right|$. A directed edge $\varepsilon_{ji} \triangleq (v_j, v_i) \in \set{E}$ indicates that node $v_j$ receives information from node $v_i$. The nodes that transmit information to node $v_j$ directly are called in-neighbors of node $v_j$, and belong to the set $\inneighbor{j}=\{v_i \in \set{V} \mid \varepsilon_{ji} \in \set{E}\}$. The number of nodes in the in-neighborhood set is called in-degree and is denoted by $\indegree{j} = \left|\inneighbor{j}\right|$. The nodes that receive information from node $v_j$ directly are called out-neighbors of node $v_j$, and belong to the set $\outneighbor{j}=\{v_l \in \set{V} \mid \varepsilon_{lj} \in \set{E}\}$. The number of nodes in the out-neighborhood set is called out-degree and is denoted by $\outdegree{j}= \left|\outneighbor{j}\right|$. Each node $v_j \in \set{V}$ has immediate access to its own local state, and thus we assume that the corresponding self-loop is available $\varepsilon_{jj} \in \set{E}$, although it is not included in the nodes' out-neighborhood and in-neighborhood. 
A directed path from $v_i$ to $v_l$ with a length of $t$ exists if a sequence of nodes $i \equiv l_0,l_1, \dots, l_t \equiv l$ can be found, satisfying $(l_{\tau+1},l_{\tau}) \in \mathcal{E}$ for $ \tau = 0, 1, \dots , t-1$. In $\mathcal{G}$ a node $v_i$ is reachable from a node $v_j$ if there exists a path from $v_j$ to $v_i$ which respects the direction of the edges. The digraph $\mathcal{G}$ is said to be strongly connected if every node is reachable from every other node.

\subsection{Distributed Average Consensus Problem}

Let us assume that at each time instant $k\in\mathbb{Z}_{\geq0}$ each node $v_j \in \mathcal{V}$ maintains a scalar\footnote{The results can be extended to the case where agents maintain vectorial states.} state $x_{j,k} \in \mathbb{R}$. In the distributed average consensus problem, the goal of the agents is to reach consensus to a value equal to the average of their initial states, which is defined as
\begin{align}\label{eq:ac_problem}
    x_{\text{ave}} := \dfrac{1}{n} \sum_{j=1}^{n} x_{j,0}.
\end{align}
Due to the absence of global knowledge at each agent, agents are required to execute an iterative distributed algorithm to eventually converge to the average consensus value, by means of local communication and computation. In particular, agents update their local states using information received from their immediate neighboring agents, through communication channels. 

In practice, however, the exchange of information is often restricted to be unidirectional (instead of bidirectional) as a consequence of diverse transmission power, interference levels, and communication ranges. Thus, although an agent $v_j$ may receive information from agent $v_i$, that does not necessarily imply that $v_j$ can also send information back to agent $v_i$. Consequently, the network topology is best modeled by a
directed graph (digraph).

\subsection{Average Consensus in Digraphs using Surplus Consensus}

A linear algorithm for achieving average consensus over directed and strongly connected networks was proposed in~\cite{cai2012average}. Each agent \( v_j \in \set{V} \) maintains a local state \( x_{j,k} \in \mathbb{R} \) and a surplus variable \( s_{j,k} \in \mathbb{R} \), initialized as \( x_{j,0} \in \mathbb{R} \), \( s_{j,0} = 0 \). At each iteration \( k \), agent \( v_j \) sends \( x_{j,k} \) and \( c_{lj}s_{j,k} \) to its out-neighbors \( v_l \in \outneighbor{j} \), where the weights \( c_{lj} \) are defined as:
\begin{align}\label{eq:c-weights}
    c_{lj}=\begin{cases}
    \frac{1}{1+\outdegree{j}}, & \text{if } v_l \in \outneighbor{j} \lor l = j,\\
    0, & \text{otherwise}.
    \end{cases}
\end{align}
Each agent \( v_j \) receives \( x_{i,k} \) and \( c_{ji}s_{i,k} \) from in-neighbors \( v_i \in \inneighbor{j} \), and updates its variables as:
\begin{subequations}\label{eq:ppac}
\begin{align}
x_{j,k+1} &= r_{jj} x_{j,k} + \gamma s_{j,k} + \sum_{v_i \in \inneighbor{j}} r_{ji} x_{i,k}, \\
s_{j,k+1} &= c_{jj} s_{j,k} + x_{j,k} - x_{j,k+1} + \sum_{v_i \in \inneighbor{j}} c_{ji} s_{i,k},
\end{align}
\end{subequations}
where \( \gamma > 0 \) is the \emph{surplus gain}\footnote{The choice of \( \gamma \) requires global knowledge of the network size.}, used to tune the convergence rate. The weights \( r_{ji} \), used to combine received \( x_{i,k} \), are given by:
\begin{align}\label{eq:r-weights}
    r_{ji} = \begin{cases}
    \frac{1}{1+\indegree{j}}, & \text{if } v_i \in \inneighbor{j} \lor j = i,\\
    0, & \text{otherwise}.
    \end{cases}
\end{align}



\begin{remark}
\emph{Push weights} \(c_{lj} \geq 0\) are assigned by the transmitting agents based on their out-degree, which is either estimated or computed (see~\cite{hadjicostis2015robust,charalambous2015distributed,makridis2023utilizing}). These weights form a column-stochastic matrix \(C \in \mathbb{R}_{\geq0}^{n \times n}\) with \(\mathbf{1}^\top C = \mathbf{1}^\top\).
\end{remark}

\begin{remark}
\emph{Pull weights} \(r_{ji} \geq 0\) are assigned by the receiving agents using their in-degree, which can be directly counted. These weights form a row-stochastic matrix \(R \in \mathbb{R}_{\geq0}^{n \times n}\) with \(R\mathbf{1} = \mathbf{1}\).
\end{remark}

The main idea behind \cite{cai2012average} is to keep the sum of agents' variables (state and surplus) time-invariant, and equal to the initial sum of agents' variable, \ie $\mathbf{1}^{\top}({\bm x}_k + {\bm s}_k) = \mathbf{1}^{\top} {\bm x}_0$ for all time $k\geq0$, where ${\bm x}_k$ and ${\bm s}_k$ are stacks of the state and surplus variables of all agents into column vectors. 



\section{Problem Formulation}\label{sec: problem}

Consider a communication network represented by a digraph $\mathcal{G} = (\mathcal{V}, \mathcal{E})$ comprising $n = |\mathcal{V}|$ nodes. 
We assume that the communication channels between nodes in our network $\mathcal{G}$ have limited bandwidth. 
At each time instant $k\geq0$ each node $v_j \in \mathcal{V}$ maintains a scalar state $x_{j,k} \in \mathbb{R}$. 
In this paper, we aim to design a distributed algorithm that enables nodes to reach consensus to a value equal to the average of their initial states defined as $x_{\text{ave}}$ in \eqref{eq:ac_problem}. 
Additionally, during the operation of our algorithm, nodes are required to communicate by exchanging quantized valued messages due to the limited bandwidth nature of the network.

\section{Quantized Average Consensus with\\ Adaptive Zooming and Midpoint Shifting}
\label{sec:algo}

Assume that each agent $v_j\in\set{V}$ is given an initial quantization step size $\Delta_0$, midpoint $\sigma_0$, and $b$ number of bits for quantizing the information that is to be sent to its out-neighbors $v_l\in\outneighbor{j}$, \ie state variable $\check{x}_{j,k}$, surplus variable $\check{s}_{j,k}$. Then each agent $v_j$ sends the following quantized messages to its out-neighbors $v_l\in\outneighbor{j}$:
\begin{subequations}\label{eq:quantized_values}
    \begin{align}
        \check{x}_{j,k} &= Q\big(x_{j,k},b,\Delta_k,\sigma_k\big),\\
        \check{s}_{j,k} &= Q \big(s_{j,k},b,\Delta_k,0\big),
    \end{align}
\end{subequations}
where $Q(\cdot)$ is identical for each node. Here it is important to note that, the quantizer that gives the quantized value $\check{s}_{j,k}$ for the surplus variable $s_{j,k}$, takes the same number of bits and quantization step size as the quantizer for $x_{j,k}$ with the only difference being the midpoint $\sigma$ which is always at $0$. 
A more detailed discussion regarding the midpoint $\sigma$ is provided in subsequent sections. 

Before introducing the zooming mechanism for adapting the quantization levels, we state the following assumptions.

\begin{assumption}\label{assum:1}
    All the agents $v_j\in\set{V}$ know an upper bound on the diameter of the network $\bar{D}\geq D$, where $D$ is the diameter of the network, defined as the longest shortest path between any two nodes (respecting the direction of the links).
\end{assumption}

\begin{assumption}\label{assum:2}
    All the agents $v_j\in\set{V}$ set their initial quantization step size to $\Delta_0$ and quantization midpoint to $\sigma_0$.
\end{assumption}

\begin{assumption}\label{assum:3}
    All the agents $v_j\in\set{V}$ are given the same surplus gain $\gamma>0$ and quantization zooming factor $\alpha>0$.
\end{assumption}

Assumption~\ref{assum:1} enables us to design max- and min-consensus protocols that are performed by the agents every $\bar{D}$ time steps. 
These protocols enable agents to take coordinated decisions for zooming out or zooming in the quantization level $\Delta_k$ of their quantizer, every $\bar{D}$ steps, with the aim of achieving faster and more accurate quantized average consensus, compared to fixed quantization strategies. 
Assumption~\ref{assum:1} is standard in the literature (see for example \cite{cady2015finite,charalambous2015distributed}). 
However, one can employ distributed mechanisms to obtain the diameter of the network either via exact computation \cite{oliva2016distributed} or estimation \cite{garin2012distributed}. 
Assumption~\ref{assum:2} enables nodes to maintain efficient communication during the operation of our algorithm over the bandwidth-limited communication links. Assumption~\ref{assum:3} is necessary for the operation of our algorithm enabling consistent coordination with quantized message exchange among nodes.



\subsection{Coordinated Quantization Zooming}
In what follows we will introduce a zooming mechanism for adapting the quantization levels for the state $x_{j,k}$ and surplus variable $s_{j,k}$. Recall that all agents have access to the same initial quantization step size $\Delta_0$ and number of bits for quantization $b$. Executing a max-consensus procedure which converges to the maximum value exchanged in the network after at most $D$ steps, agents are able to coordinate their zooming mechanism for adjusting the quantization step size $\Delta_k$ accordingly. In essence, we would like agents to zoom out their quantizers (increase $\Delta_k$) when at least one agent's value is out of range. Conversely, if all agents' values are within the range, they should either zoom in their quantizers (decrease $\Delta_k$) or keep them constant (keep $\Delta_k$). The quantization step size $\Delta_k$ is updated every $\bar{D}$ consensus iterations and remains fixed during each $\bar{D}$-step interval.

Before we formally define the coordination mechanism for quantization zooming, we let $\set{K}_{\bar{D}}$ denote the set of times when agents adjust their quantization step size, \ie
\begin{align}
    \set{K}_{\bar{D}} = \{ k \in \mathbb{Z}_{\geq 0} \mid k = \bar{D}h, \, h \in \mathbb{Z}_{\geq0} \}.    
\end{align}
Adjusting the quantization step size requires agents to coordinate in order to agree on their quantizers' operation (\eg zoom-in or zoom-out), and its adjustment on the step-size (\ie quantization precision). To adjust the quantization step size $\Delta_k$, each agent $v_j \in \set{V}$ maintains a variable ${\zeta}_{j,k}$ which is (re)set at each time $k \in \set{K}_{\bar{D}}$, as follows:
\begin{align}\label{eq:zeta}
    \zeta_{j,k} &= \begin{cases}
        \hfill 1, &  \text{if } x_{j,k} > \sigma_k + \bar{x}_k \lor x_{j,k} < \sigma_k - \bar{x}_k,\\
        \hfill -1, &  \text{if } \frac{\sigma_k - \bar{x}_k}{1+\alpha} < x_{j,k} < \frac{\sigma_k + \bar{x}_k}{1+\alpha},\\        
        \hfill 0, & \text{otherwise},
    \end{cases}
\end{align}
where $\alpha>0$ is the zooming factor. The variable $\zeta_{j,k}$ captures how well the value to be quantized fits within the current quantization range and classifies this alignment into three distinct regions: central (well within range), peripheral (near the range boundaries), and out-of-range (exceeding the range), thereby guiding dynamic adjustments to the quantization scale. For time instants where $k \notin \set{K}_{\bar{D}}$, each agent $v_j \in \set{V}$ updates its local variable $\zeta_{j,k}$ by means of local computation and communication, through a max-consensus procedure, as follows:
\begin{align}\label{eq:max_consensus}
        \zeta_{j,k+1} = \max_{v_i \in \inneighbor{j}\cup \{v_j\}}\big\{\zeta_{i,k}\big\}.
    \end{align}
By the end of ${\bar{D}}$ iterations of the above max-consensus procedure, agents are guaranteed to agree on the maximum value in the network \cite{giannini2013convergence}, \ie
\begin{align}\label{eq:max_zeta}
    \zeta_k &= \max_{v_j \in \set{V}} \big\{ \zeta_{j,k} \big\}.
\end{align}

Based on the above coordination scheme, at every $k\in\set{K}_{\bar{D}}$, agents synchronously adjust their quantization step-size for the next ${\bar{D}}$ steps, as
\begin{align}\label{eq:stepsize}
    \Delta_{k+1} = 
    \begin{cases}
        \hfill (1+\alpha)\Delta_k, & \text{if}~\zeta_k=1,\\
        \hfill \Delta_k/(1+\alpha), & \text{if}~\zeta_k=-1,\\
        \hfill \Delta_k, & \text{otherwise.}
    \end{cases}
\end{align}
Note that, for $k\notin \set{K}_{\bar{D}}$, agents only exchange and update their variables $\zeta_{j,k}$, respecting the communication topology, as in \eqref{eq:max_consensus}, while they hold the last updated quantization step size $\Delta_k$ unchanged, until the next $k\in\set{K}_{\bar{D}}$.

\subsection{Coordinated Quantization Midpoint Shifting}
To achieve arbitrarily large reduction of the average consensus error, we need to employ a midpoint shifting mechanism for adjusting the symmetry point of the agents' quantizers. In essence, we would like to design an update rule for the quantizers' midpoint value $\sigma_k$ that will track the actual average consensus value of the network in a distributed fashion. This will allow agents to further zoom in their quantizers such that they reach average consensus in higher precision. 

For shifting the quantizers' midpoint $\sigma_k$ in a distributed and coordinated manner, each agent $v_j \in \set{V}$ maintains two auxiliary variables $M_{j,k}$ and $\mu_{j,k}$ which are updated using max- and min-consensus, respectively. These variables are reset at each time $k \in \set{K}_{\bar{D}}$, as $M_{j,k}=\check{x}_{j,k}$ and $\mu_{j,k}=\check{x}_{j,k}$, and they are updated at $k \notin \set{K}_{\bar{D}}$ as follows:
\begin{subequations}\label{eq:max-min-consensus}
\begin{align}
    M_{j,k+1} &= \max_{v_i \in \inneighbor{j}\cup \{v_j\}}\big\{M_{i,k}\big\},\\
    \mu_{j,k+1} &= \min_{v_i \in \inneighbor{j}\cup \{v_j\}}\big\{\mu_{i,k}\big\}.
\end{align}
\end{subequations}
By the end of ${\bar{D}}$ iterations of the aforementioned max- and min-consensus procedures, agents compute the maximum and minimum quantized value in the network, \ie $M_{k}=\max_{v_j \in \set{V}}\{M_{j,k}\}$ and $\mu_{k}=\min_{v_j \in \set{V}}\{\mu_{j,k}\}$, respectively. Based on these values, they shift their quantizers' midpoint according to:
\begin{align}\label{eq:shifting}
    \sigma_k = \frac{1}{2}(M_k + \mu_k).
\end{align}

\subsection{Quantized Average Consensus}
At each time step $k\in\mathbb{Z}_{\geq0}$ and upon the quantization of the state and surplus variables with the coordinated alternating zooming, each agent $v_j$ sends to its out-neighbors the values in \eqref{eq:quantized_values}, and updates its variables with the corresponding quantized values it received from its in-neighbors, as follows:
\begin{subequations}
\begin{align}
x_{j,k+1} &= x_{j,k} + \gamma s_{j,k} + \!\!\! \sum_{v_i \in \inneighbor{j}} r_{ji}\check{x}_{i,k},\\
s_{j,k+1} &= s_{j,k} + x_{j,k} - x_{j,k+1} + \!\!\! \sum_{v_i \in \inneighbor{j}} c_{ji}\check{s}_{i,k},
\end{align}
\end{subequations}
where $c_{ji}$ and $r_{ji}$ are given in \eqref{eq:c-weights} and \eqref{eq:r-weights}, respectively. Note that, the quantized value $\check{s}_{i,k}$ is sent over the communication channel along with the out-degree of agent $v_i$, $d_i^{\texttt{out}}\in\mathbb{N}$. 
Hence the weight $c_{ji}$ can be reconstructed exactly by agent $v_j$, by simply setting $c_{ji}=(d_i^{\texttt{out}}+1)^{-1}$.

The compact form of the proposed \ouralgorithm{} algorithm can be written as:
\begin{subequations}\label{eq:qac_az}
\begin{align}
{\bm x}_{k+1} &= {\bm x}_k + \gamma{\bm s}_k + (R-I) \check{\bm {x}}_k,\label{eq:qac_az_1}\\
{\bm s}_{k+1} &= {\bm x}_k - {\bm x}_{k+1} + {\bm s}_k + (C-I) \check{\bm s}_k\label{eq:qac_az_2}.
\end{align}
\end{subequations}

\begin{remark}
    Our proposed algorithm is a distributed average consensus variant of the Push-Pull Gradient Tracking method introduced by Song \emph{et al.}~\cite{song2022compressed}, which addresses distributed optimization with compression over general directed networks. Although the algorithm in \cite{song2022compressed} can solve the average consensus problem as a special case, it is specifically designed for optimizing global objective functions with local information at each node. As a result, it introduces unnecessary computational and communication overhead for consensus tasks, such as maintaining additional variables, performing gradient evaluations, and communicating them over the network. In contrast, our algorithm directly targets distributed average consensus using a dynamic quantization scheme that achieves arbitrarily small error with a fixed number of bits.
\end{remark}

\begin{lemma}\label{lemma:1}
    The total mass in the network (hence the state average) is preserved using the iterations in \eqref{eq:qac_az} for all $k\in\mathbb{Z}_{\geq0}$, \ie $\mathbf{1}^{\top}\big( {\bm x}_{k} + {\bm s}_{k} \big) = \mathbf{1}^{\top} {\bm x}_{0}$. 
\end{lemma}
\begin{proof}
    Consider the summation of all agents' $x$ and $s$ values at time step $k+1$, \ie $\mathbf{1}^{\top}\big( {\bm x}_{k+1} + {\bm s}_{k+1} \big)$. Then, plugging  \eqref{eq:qac_az_1} into \eqref{eq:qac_az_2} and substituting into the evolution at time step $k+1$, we get:
    \begin{align}
        \mathbf{1}^{\top}\big( {\bm x}_{k+1} + {\bm s}_{k+1} \big) \nonumber = & \mathbf{1}^{\top} {\bm x}_{k} + \mathbf{1}^{\top} {\bm s}_{k} + \mathbf{1}^{\top} C \check{{\bm s}}_{k} - \mathbf{1}^{\top} I \check{{\bm s}}_{k} \nonumber\\ = &
        \mathbf{1}^{\top} \big( {\bm x}_{k} + {\bm s}_{k} \big),
    \end{align}
    where the last equality comes from the column-stochasticity of matrix $C$. This proves the statement of the lemma.
\end{proof}

Substituting \eqref{eq:qac_az_1} into \eqref{eq:qac_az_2}, and augmenting the state by letting ${\bm z}_k = \begin{bmatrix}{\bm x}_k^{\top}, {\bm s}_k^{\top}\end{bmatrix}^{\top}\in\mathbb{R}^{2n}$ and $\check{\bm z}_k = \begin{bmatrix}\check{\bm x}_k^{\top}, \check{\bm s}_k^{\top}\end{bmatrix}^{\top}\in\mathbb{R}^{2n}$ the iterations of the proposed \ouralgorithm{} can be rewritten in the following form: 
\begin{align}\label{eq:aug_form}
{\bm z}_{k+1}  = 
\begin{bmatrix}
    I_n & \gamma I_n\\
    0 & (1-\gamma)I_n
\end{bmatrix} {\bm z}_{k} + \begin{bmatrix}
    R-I_n & 0_{n\times n}\\
    I_n-R & C-I_n
\end{bmatrix} \check{\bm z}_{k}.
\end{align}

Now, considering the quantization error we can define $\check{\bm z}_{k} \triangleq {\bm z}_k + {\bm e}_k$, to further obtain
\begin{align}\label{eq:aug_form_error}
    {\bm z}_{k+1} = (\Gamma + \Pi) {\bm z}_{k} + (\Pi - I_{2n}){\bm e}_k,
\end{align}
where $\Gamma,\Pi\in\mathbb{R}^{2n\times 2n}$, with
\begin{align}
    \Gamma = \begin{bmatrix} 0 & \gamma I_n \\ 0 & - \gamma I_n \end{bmatrix}, \quad \text{and} \quad \Pi = \begin{bmatrix} R & 0_{n\times n} \\ I_n-R & C \end{bmatrix}.
\end{align}

\begin{theorem}\label{theorem:1}
    The push-pull quantized average consensus algorithm in \eqref{eq:qac_az} achieves average consensus, \ie ${\bm x}_{k} \rightarrow x_{\text{ave}}\mathbf{1}$ and ${\bm s}_{k} \rightarrow \mathbf{0}$ as $k\rightarrow\infty$, with a properly chosen surplus gain $\gamma>0$ and sufficiently small zooming factor $\alpha>0$.
\end{theorem}

\begin{proof}[\textbf{Sketch of the proof of Theorem~\ref{theorem:1}}]\label{sec:appendix}
We analyze the convergence of the iteration by examining the behavior of the quantization error vector ${\bm e}_k$. In the ideal case where ${\bm e}_k = \mathbf{0}$ for all $k$, the update reduces to a known linear iteration from \cite{cai2012average}, which converges to the correct consensus value. For the case of constant quantization error, we show that the system evolves as a linear time-invariant system with input $(\Pi - I){\bm e}$, and the fixed point is well-defined due to the spectral properties of the matrix $(\Gamma + \Pi)$. Specifically, its dominant eigenvalue is $1$, with a corresponding left eigenvector orthogonal to the error input, and all other eigenvalues lie strictly inside the unit circle. This structure ensures convergence to a unique fixed point. Finally, when the quantization error ${\bm e}_k$ is time-varying but bounded, which is ensured due to the zooming and shifting mechanisms, the system's dynamics are perturbed by a vanishing input. This ensures that the error vanishes asymptotically, and the iterates converge to the correct consensus value. The convergence is guaranteed by mass preservation and the design of the quantization parameters.
\end{proof}


\section{Simulation Results}
\label{sec:simulations}

In this section, we illustrate the behavior of the proposed distributed algorithm and we showcase its performance. Throughout the simulations, we consider a fixed directional network of $n=5$ agents, which we model as a directed graph $\set{G}=(\set{V},\set{E})$, of diameter $D=4$. Under this setup, each agent $v_j\in\set{V}$ knows an upper bound on the diameter of the network, ${\bar{D}}=4$, and initializes its local variables with $\Delta_0=0$, $\sigma_0=0$, $s_{j,k}=0$, $\gamma=0.2$, and $\alpha\in\{1.2,5\}$. Note that, we omit the agents' index $j$ for the variables that are always the same between agents due to either the coordination mechanisms we described in the previous sections, \ie $\Delta_k$ and $\sigma_k,$ and based on the assumption that agents have global information of the parameters $\alpha, \gamma,$ and ${\bar{D}}$, apriori. The initial state variable of agent $v_j$, $x_{j,0}$, is assumed to be randomly initialized from a uniform distribution in the interval $[0,1000]$. In general, to avoid saturation and ensure full coverage of the initial values that $x_{j,0}$ can take between the specified interval, then, agents need at least $10$ bits, \ie $2^{10}>1000$. For a lower number of bits it is not guaranteed that values will be representable, and hence, some agents may experience saturation on their quantized values. 
Although the methods in the literature proposed solutions where the number of bits for quantization is enough to represent the information that is to be exchanged, the \ouralgorithm{} algorithm of this work can handle cases where the number of quantization bits are not enough to fully represent a real value, in the expense of more communication rounds. In what follows, we illustrate the evolution of the values of agents in the network, and we show the convergence speed of the agents for different settings. 

In Fig.~\!\ref{fig:alpha_small} and Fig.~\!\ref{fig:alpha_large} we present the evolution of the values of the agents executing the proposed algorithm, over the network and initial values we described before, and we examine the average consensus error captured by $\lVert {\bm x}_k - \mathbf{1}x_{\text{ave}}\rVert_2$, for $\alpha=1.2$ and $\alpha=5$, respectively. As shown in these figures, agents converge faster to the average consensus value when $\alpha=5$ rather than $\alpha=1.2$. At the first rounds of communication for $b\in\{3,8\}$, agents increase their quantization step sizes $\Delta_k$ every ${\bar{D}}=4$ iterations, since the values of their state variables $x_{j,k}$ cannot be fully represented by $b\in\{3,8\}$ bits. Hence, they increase their step size $\Delta_k$, and shift their midpoint $\sigma_k$ in a distributed and coordinated manner, up to the point where the values they want to send to their out-neighbors are within their quantization range. Clearly, for $b=24$, agents start by zooming-in their quantizers (by decreasing the quantization step size) since with $24$ bits, all their values can be represented adequately. Notice that, for a relatively lower number of bits and higher value of $\alpha$, the quantizers' step size fluctuates since the zooming factor $\alpha$ is too big for the given number of bits.  

\begin{figure}
    \centering
    \includegraphics[width=0.94\linewidth]{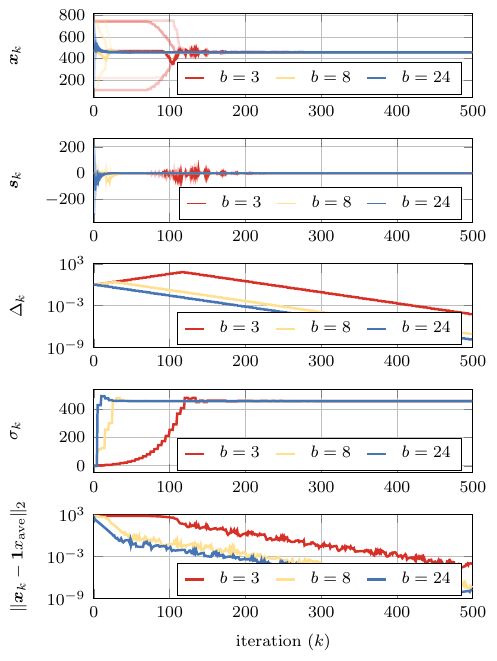}
    \vspace{-10pt}
    \caption{Evolution of the \ouralgorithm{} algorithm over a directed network of $n=5$ agents with $D=4$ and $\alpha=1.2$.}
    \label{fig:alpha_small}
\end{figure}

\begin{figure}
    \centering
    \includegraphics[width=0.94\linewidth]{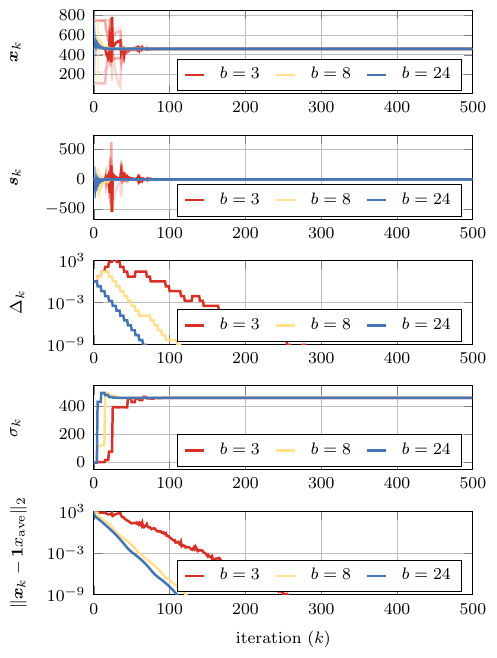}
    \vspace{-10pt}
    \caption{Evolution of the \ouralgorithm{} algorithm over a directed network of $n=5$ agents with $D=4$ and $
    \alpha=5$.}
    \label{fig:alpha_large}
\end{figure}

\subsection*{Convergence Time - Preliminary Numerical Analysis}


We now analyze the performance of the \ouralgorithm{}  algorithm in terms of average convergence time and the number of bits used for quantization. Specifically, we conduct $50$ simulations on a network with $n=5$ agents, and upper bound on the network diameter ${\bar{D}}=4$, as shown in Fig.~\!\ref{fig:bits_convergence_net1}. In these simulations, we try different quantizer zooming factors $\alpha\in\{1.2,1.3,1.5,1.8,2,10\}$ and quantization bits $b\in\{2,4,6,8,10,12\}$ for all agents, and record the average convergence time, \ie $\min\big\{k \mid \lvert x_{i,k}-x_{j,k}\rvert\leq 10^{-8}, \forall v_j, v_i \in \mathcal{V}\big\}$.  
%
%
Observe that, by increasing the number of bits for communication, the convergence time of the \ouralgorithm{} algorithm decreases significantly, especially with smaller zooming factor $\alpha$. However, notice that, selecting a higher quantization zooming factor, such as $\alpha=10$, leads to the saturation of the \ouralgorithm{} algorithm, preventing further improvement in quantizer precision. This saturation arises from excessive fluctuations in the quantization step size, which inhibit the zooming mechanism from effectively reducing it. Thus, in these cases the algorithm cannot reach $\lvert x_{i,k}-x_{j,k}\rvert\leq 10^{-8}, \forall v_j, v_i \in \mathcal{V}$, and hence the average convergence time is not shown in Fig.~\!\ref{fig:bits_convergence_net1} for $b\in\{2,4,6\}$.

\begin{figure}[t]
    \centering
    \includegraphics[width=0.9975\linewidth]{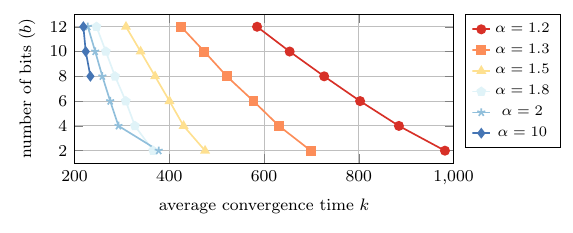}
    \vspace{-25pt}
    \caption{Average convergence time of the \ouralgorithm{} algorithm, \ie $\min\{k \mid \lvert x_{i,k}-x_{j,k}\rvert\leq 10^{-8}, \forall v_j, v_i \in \mathcal{V}\}$  in a $5$-agent directed network with average diameter $D=4$ and $\gamma=0.2$.}
    \label{fig:bits_convergence_net1}
\end{figure}


\section{Conclusions and Future Directions}


In this paper, we addressed the problem of achieving quantized average consensus in possibly unbalanced directed graphs while operating under limited bandwidth constraints. Existing approaches largely rely on probabilistic convergence, suffer from residual quantization errors, or require unbounded message sizes, making them impractical for many real-world applications. To overcome these limitations, we proposed a novel distributed average consensus algorithm, herein called the Push-Pull Average Consensus algorithm with Dynamic Compression ({\ouralgorithm}) algorithm, that incorporates an adaptive quantization mechanism. Our approach enables agents to deterministically reach consensus on the exact average without requiring global knowledge of the network. We also provided a preliminary numerical convergence analysis and demonstrated the effectiveness of our method through simulations.


Despite the promising results, several avenues remain open for future research. One key direction is to develop distributed stopping mechanisms with which agents can stop communicating provided the network has reached consensus within a predefined distance from the exact average. Finally, experimental validation on real-world distributed systems would help assess the algorithm’s robustness in practical scenarios.

\bibliographystyle{IEEEtran}
\bibliography{references}

\end{document}